\documentclass[12pt]{article}

\usepackage{setspace}
\usepackage{authblk}
\usepackage[top=1in, bottom=1in, left=1in, right=1in]{geometry}

\usepackage{amsmath}
\usepackage{graphicx}
\usepackage{natbib}
\usepackage{url} % not crucial - just used below for the URL 
\usepackage{ulem}

\usepackage{mathtools} % amsmath with fixes and additions
\usepackage{booktabs} % commands to create good-looking tables

\usepackage{tikz}
\usetikzlibrary{matrix,decorations.pathreplacing, calc, positioning,fit, arrows.meta}

\usepackage{graphicx}
\usepackage{amsmath}
\usepackage{amssymb}
\usepackage{amsthm}
\usepackage{graphicx}
\usepackage{epstopdf}
\usepackage{comment}
\usepackage{subcaption}
\usepackage{esint}
\usepackage{bm}
\usepackage{latexsym}

\usepackage[linesnumbered,ruled,vlined]{algorithm2e}
\SetKwInput{kwInit}{Initialize}

\DeclareMathOperator*{\argmin}{arg\,min}

\newcommand{\bA}{\mathbf{A}}

\newcommand{\bx}{\bm{x}}
\newcommand{\by}{\bm{y}}
\newcommand{\bu}{\bm{u}}

\newcommand{\bv}{\bm{v}}

\def\trans{^{\rm T}}

\newcommand{\bI}{\mathbf{I}}
\newcommand{\bG}{\mathbf{G}}

\newcommand{\bB}{\mathbf{B}}
\newcommand{\bC}{\mathbf{C}}
\newcommand{\bD}{\mathbf{D}}
\newcommand{\bE}{\mathbf{E}}
\newcommand{\bJ}{\mathbf{J}}

\newcommand{\bM}{\mathbf{M}}
\newcommand{\bP}{\mathbf{P}}
\newcommand{\bX}{\mathbf{X}}
\newcommand{\bY}{\mathbf{Y}}
\newcommand{\bU}{\mathbf{U}}
\newcommand{\bV}{\mathbf{V}}
\newcommand{\bW}{\mathbf{W}}
\newcommand{\bZ}{\mathbf{Z}}
\newcommand{\bR}{\mathbf{R}}
\newcommand{\bQ}{\mathbf{Q}}
\newcommand{\bT}{\mathbf{T}}
\newcommand{\bS}{\mathbf{S}}

\newcommand{\bzeta}{\boldsymbol{\zeta}}

\newcommand{\bSigma}{\boldsymbol{\Sigma}}

\newcommand{\bPi}{\boldsymbol{\Pi}}
\newcommand{\bPhi}{\boldsymbol{\Phi}}

\newcommand{\bzero}{\mathbf{0}}

\def\wh{\widehat}
\def\wt{\widetilde}
\usepackage{mathabx}

\def\trans{^{\rm T}}

\newcommand{\norm}[1]{\Vert#1\Vert}
\newcommand{\Norm}[1]{\left\Vert#1\right\Vert}

\newcommand{\abs}[1]{\vert#1\vert}

\usepackage{xcolor}

\newcommand{\expectation}{\mathbb{E}}
\newcommand{\probability}{\mathbb{P}}

\newtheorem{remark}{Remark}
\newtheorem{assumption}{Assumption}
\newtheorem{theorem}{Theorem}

\title{Mixed Matrix Completion in Complex Survey Sampling under Heterogeneous Missingness}
\author[1]{Xiaojun Mao}
\author[2]{Hengfang Wang \thanks{
E-mail addresses: maoxj@sjtu.edu.cn, hengfang@fjnu.edu.cn, wangzl@xmu.edu.cn, syang24@ncsu.edu}}
\author[3]{Zhonglei Wang}
\author[4]{Shu Yang}
\affil[1]{School of Mathematical Sciences,
Ministry of Education Key Laboratory of Scientific and Engineering Computing, 
Shanghai Jiao Tong University}
\affil[2]{School of Mathematics and Statistics \& Fujian Provincial Key Laboratory of Statistics and Artificial Intelligence, Fujian Normal University}
\affil[3]{Wang Yanan Institute for Studies in Economics and School of Economics,
Xiamen University}
\affil[4]{Department of Statistics, North Carolina State University}
\date{}

\begin{document}

\maketitle

\begin{abstract}
  Modern surveys with large sample sizes and growing mixed-type questionnaires require robust and scalable
  analysis methods. In this work, we consider recovering a mixed dataframe matrix, obtained by complex survey sampling, with entries following different canonical exponential distributions and subject to heterogeneous missingness. To tackle this challenging task, we propose a two-stage procedure: in the first stage, we model the entry-wise missing mechanism by logistic regression, and in the second stage, we complete the target parameter matrix by maximizing a weighted log-likelihood with a low-rank constraint. We propose a fast and scalable estimation algorithm that achieves sublinear convergence, and the upper bound for the estimation error of the proposed method is rigorously derived.
  Experimental results support our theoretical claims, 
 and the proposed estimator shows its merits compared to other existing methods. 
The proposed method is applied to analyze the National Health and Nutrition Examination Survey data.
\end{abstract}

\section{Introduction}

Survey sampling is a touchstone for social science \citep{elliott2017inference,haziza2017construction}. Modern technologies have accelerated the sampling speed with more mixed-type questionnaires, e.g., National  Health and Nutrition Examination Survey  and 
web-based surveys \citep{rivers2007sampling}.  However, in practice, non-response is ubiquitous in survey data with arbitrary missingness patterns. If the missing mechanism is informative, ignoring the missing values leads to biased estimation.  More importantly, a large survey with many mixed-type questionnaires also requires more scalable and robust methods.

We envision the survey dataframe as a matrix, where rows correspond to subjects and columns to responses to different questions whose entries suffer from missingness.  Imputation methods are commonly used to address such
missingness. Generally, imputation methods fall into two categories: row-wise imputation and column-wise imputation. Multiple imputation \citep{rubin1976inference} is a popular row-wise imputation method and it leverages a posterior predictive distribution given the observations to impute the missing ones for each row. Multiple imputation requires positing a joint distribution, which is stringent in practice, and it is computationally heavy or even infeasible with a growing number of questions. On the other hand, hot-deck imputation \citep{chen2000nearest,kim2004fractional} is a typical column-wise imputation method and imputes missing values by observations in the same column with a predetermined distance metric. However, it is unrealistic to use the same metric to impute all the missing values, especially for large surveys with many mixed-type questions. Furthermore, due to the arbitrarily missing pattern, some donor pools may be limited.

As their names manifest, row-wise and column-wise imputation methods utilize only partial information, so do not harvest the full information of the data matrix. In contrast, matrix completion methods \citep{candes2009exact,keshavan2009matrix,mazumder2010spectral,koltchinskii2011nuclear,negahban2012restricted,fan2019177} leverage matrix structures, such as low-rankness, to impute missing values simultaneously. Additionally, low-rankness is naturally present for survey data with block-wise questionnaires and common sampling designs such as stratified and cluster sampling. 
 Mixed-type responses in large surveys require matrix completion for mixed data frames \citep{kiers1991simple,pages2014multiple,udell2016generalized}. In this vein, \cite{robin2020main} studied the main effects and interactions with low-rankness and sparse matrix completion. \cite{alaya2019collective} studied collective matrix completion whose entries come from exponential family distributions. Some survey variables are fully observed and can serve as side information to improve estimation efficiency. Inductive matrix completion \citep{xu2013speedup,jain2013provable,wang2023multitask} modeled side information by the matrix factorization method. \cite{fithian2018flexible} leveraged row-wise and column-wise side information for reduced-rank modeling of matrix-valued data. \cite{chiang2018using} studied matrix completion with missing and corrupted side information. \cite{mao2019matrix,mao2019matrixarxiv} proposed matrix completion with covariates using the column-space decomposition method. Meanwhile, most of the matrix completion literature treats the missing scheme as uniform missing, that is, missing completely at random. In the vein of missing data literature, heterogeneous missingness is more realistic and has been well studied \citep{little2019statistical,kim2021statistical},  and the missing mechanism is usually modeled by logistic regression on the covariates. Without covariate information, it is difficult to model heterogeneous missingness for matrix completion \citep{mao2018matrix}.

This paper considers recovering a mixed survey data matrix with entries following canonical exponential distributions given auxiliary information and subject to heterogeneous missingness. To tackle this challenging problem, we propose a two-stage procedure: in the first stage, we employ logistic regression to model the entry-wise missing mechanism with auxiliary variables, and in the second stage, we leverage an inverse probability weighted pseudolikelihood with a low-rankness constraint for matrix recovery, where the estimated response rates serve as surrogate probabilities.
 We establish the statistical guarantee of the proposed method and present the upper bound of the estimation error. Computationally, we adopt a fast iterative shrinkage-thresholding algorithm (FISTA) \citep{beck2009fast} for estimation and show that it enjoys a sublinear convergence rate. To support our theoretical analysis, a synthetic experiment is conducted and the proposed estimator is shown to have its merits compared to other competitors. 
 We apply the proposed method to recover mixed-type missing values in the 2015-2016 National Health and Nutrition Examination Survey data.

Notation. For a matrix $\bS=(s_{ij}) \in \mathbb{R}^{n_{1} \times n_{2}}$, its singular values are $\sigma_{1}(\bS),\ldots, \sigma_{r}(\bS)$ in descending order. We denote the Frobenius norm of $\bS$ by $\Norm{\bS}_{F} = ( \sum_{i=1}^{n_{1}}\sum_{j=1}^{n_{2}} s_{ij}^2  )^{1/2}$, the operator norm by $\norm{\bS} = \sigma_{\max}(\bS)=\sigma_{1}(\bS)$, the nuclear norm $\norm{\bS}_{\ast} = \sum_{i=1}^{r }\sigma_{i}(\bS)$ and the sup norm $\norm{\bS}_{\infty} = \max_{i,j} \{ s_{ij}\}$.  For a positive integer $n$, define $[\![n] \!] = \{1, \ldots, n\}$.

\section{Methodology}

\subsection{Problem Formulation}

We consider a finite population containing $N$ subjects represented as $\mathcal{F} = \{(\bx_i,\by_i):i \in [\![N] \!] \}$, where $\bx_{i} \in \mathbb{R}^{D}$ is a $D$-dimensional covariate, and $\by_i=(y_{i1},\ldots,y_{iL})\trans$ is the response of interest to $L$ questions.  The dimension of $\by_i$ is denoted by $L$. We consider a sampling design without replacement.  Denote $I_i$ to be the sampling indicator of the subject $i$,  where $I_i=1$ if the subject $i$ is sampled and $0$ otherwise. 
The first-order inclusion probability for the $i$-th subject is denoted by $\pi_i = \expectation(I_i)$, where the expectation corresponds to the sampling process.
Denote $\Pi=\{\pi_1,\ldots,\pi_N\}$ to be the known first-order inclusion probabilities.
For simplicity, we assume that the first $n$ subjects are included in the sample. The covariate matrix $\bX= (x_{ij})\in \mathbb{R}^{n\times D}$ is then formed, where $x_{ij}$ represents the $j$-th covariate for the $i$-th subject. This covariate matrix contains demographic information and typically has fixed $D$ columns. Furthermore, for a questionnaire with $L$ questions, we construct an $n\times L$ matrix denoted as $\bY=(y_{ij})\in\mathbb{R}^{n\times L}$. Each entry $y_{ij}$ in this matrix corresponds to the answer provided by the  subject $i$ to the $j$-th question. Due to nonresponses, the response probability of $\bY$ is often low. To this end, we use the corresponding missing indicator matrix $\bR=(r_{ij}) \in \mathbb{R}^{n\times L}$, where for any $(i,j)\in [\![ n ]\!] \times [\![ L ]\!]$, the value of $r_{ij}$ is 1 if $y_{ij}$ is observed (not missing), and 0 otherwise.

In a survey questionnaire, the questions are usually grouped into different categories based on the types of possible answers they can have.  For instance, questions that elicit responses like "yes" or "no," or "like" and "dislike" fall under the category of binary responses. Questions that require answers in the form of nonnegative integers, such as the number of household members or pets, belong to the category of nonnegative integer responses. Meanwhile, questions that prompt continuous values like household income or odometer readings are classified as continuous responses. Assume that there are a total of $S$ categories of questions in a specific questionnaire, and denote the number of questions within the $s$-th category as $m_{s}$. Naturally, the total number of questions in the survey is given by $L = \sum_{s=1}^{S}m_{s}$.
For the sake of clarity and with a slight abuse of notation, we arrange the responses for each subject $i$ in a concatenated manner as $\by_{i} = ((\by_{i}^{(1)})\trans, \ldots, (\by_{i}^{(S)})\trans )\trans$. Here, $\by_{i}^{(s)}\in \mathbb{R}^{m_{s}}$ represents the responses to the questions of the $i$-th subject in the $s$-th category, for $s \in [\![ S ]\!]$. Consequently, the full response matrix $\bY$ can be formulated as $\bY = [\bY^{(1)}, \ldots, \bY^{(S)}]$, with $\bY^{(s)} = (y_{ij}^{(s)}) \in \mathbb{R}^{n \times m_{s}}$ containing the responses to the questions in the $s$-th category.
Following a similar concatenation approach for $\bY$, we denote the missing indicator matrix as $\bR$ to account for missing responses in the data. For each $s\in [\![ S ]\!]$, we define $\bR = [\bR^{(1)}, \ldots, \bR^{(S)}]$, where $\bR^{ (s)} = (r_{ij}^{ (s)}) \in \mathbb{R}^{n\times m_{s}}$. Each entry $r_{ij}^{(s)}$ in $\bR^{(s)}$ indicates whether the corresponding response $y_{ij}^{(s)}$ in $\bY^{(s)}$ is observed (1) or missing (0).
Traditional matrix completion methods assume a specific distribution for the entire response matrix, limiting their ability to handle mixed-type responses. To overcome this limitation, we introduce an exponential family approach, which can flexibly deal with diverse types of responses \citep{alaya2019collective,robin2020main,wang2023multitask}. Specifically, we assume that all entries within the same response category come from the same generic exponential family. 
Mathematically, 
we consider the conditional density function within the $s$-th category as
introduce the exponential family to handle different types of responses,
\begin{align*}%\label{eq: csd model main}
 f^{(s)}(y_{ij}^{(s)}\mid z_{ij}^{\ast(s)})= h^{(s)}(y_{ij}^{(s)})\exp\left\{y_{ij}^{(s)}z_{ij}^{\ast(s)} -g^{(s)}(z_{ij}^{\ast(s)} )\right\},
\end{align*}
for $(s,i,j) \in [\![S]\!]\times  [\![n]\!] \times [\![m_{s}]\!]$, where $\{(g^{(s)}, h^{(s)}): s \in  [\![ S ]\!] \}$ contains $S$ doublets of functions and $\bZ^{\ast} = [\bZ^{\ast(1)}, \ldots, \bZ^{\ast(S)}]$ is the parameter matrix with $\bZ^{\ast (s)} = (z_{ij}^{\ast (s)}) \in \mathbb{R}^{n\times m_{s}}$. 
In the context of our investigation, we posit that $\bZ^{\ast}$ exhibits low-rank characteristics. For illustrative purposes, let us consider the example of the exponential family.  
Some commonly used distributions from the exponential family are listed in Table~\ref{tab:exponential}.

\begin{table}[!t]
   \renewcommand{\arraystretch}{1.3}
        \caption{Commonly used distributions from the exponential family}
    \label{tab:exponential}
 \centering
    \begin{tabular}{ccc}
    \hline
   \bfseries  Exponential family  &  $g(z)$ & $h(y)$     \\
    \hline
 \bfseries  Bernoulli & $\log\{1 + \exp(z)\}$ & 1\\
  \bfseries Poisson &  $\exp(z)$ & $1/(y!)$ \\
  \bfseries   Gaussian & $\sigma^{2}z^{2}/2$ &
   $ (2\pi \sigma^{2})^{-1/2}\exp{ \{ -y^{2}/(2\sigma^{2})  \}}$\\
  \bfseries  Exponential & $-\log(-z)$   &   1  \\
    \hline
    \end{tabular}
\end{table}

\begin{remark}
Our setup differs from existing ones in the following aspects. Traditional matrix completion problem mainly assumes independence among subjects, but they are usually correlated under complex sampling designs. Additionally, under survey sampling, there exists another parameter matrix $\bZ_p^\ast\in\mathbb{R}^{N\times L}$ on the population level, and the parameter matrix $\bZ^\ast$ consists of rows of $\bZ_p^\ast\in\mathbb{R}^{N\times L}$ with respect to the sample $i\in [\![n]\!]$. Furthermore, if the sampling design is informative, traditional matrix completion techniques, such as \citet{alaya2019collective}, will lead to biased recovery of $\bZ^\ast$ since the sampling information is overlooked; see \citet{pfeermann1996} for details.
\end{remark}
Our goal is to estimate the parameter matrix $\bZ^{\ast}$ based on available information. 
To this end, given the density functions $\{f^{(s)}: s \in [\![S]\!]\}$,  the negative quasi-log-likelihood function for $\bZ^{\dagger}$ is 
\begin{align}\label{original likelihood}
\ell_{0}\left(\bZ^{\dagger}\right) =& -\sum_{s=1}^{S} \sum_{(i,j)\in [\![n]\!] \times [\![m_{s}]\!]   }\frac{ r_{ij}^{(s)}}{\pi_{i}} \log\left\{f^{(s)}(y_{ij}^{(s)}|z_{ij}^{\dagger(s)})\right\}\notag\\
=& \sum_{s=1}^{S} \sum_{(i,j)\in [\![n]\!] \times [\![m_{s}]\!]   }\frac{r_{ij}^{(s)}}{\pi_{i}} \left\{-y_{ij}^{(s)}z_{ij}^{\dagger(s)}+g^{(s)}\left(z_{ij}^{\dagger(s)}\right)\right\}
+ \sum_{s=1}^{S} \sum_{(i,j)\in [\![n]\!] \times [\![m_{s}]\!] } \frac{r_{ij}^{(s)}}{\pi_{i}} \log\left\{h^{(s)}(y_{ij}^{(s)}) \right\},
\end{align}
where $\bZ^{\dagger} = [\bZ^{\dagger(1)}, \ldots, \bZ^{\dagger(S)}]$ and $\bZ^{\dagger(s)} = (z_{ij}^{\dagger(s)}) \in \mathbb{R}^{n \times m_{s}}$ for $s \in [\![ S ]\!]$. In (\ref{original likelihood}), we have incorporated the inclusion probabilities to obtain the Horvitz-Thompson estimator \citep{HTestimator1952}. The Horvitz-Thompson estimator is commonly used to estimate population parameters under complex survey sampling. It adjusts for the unequal probabilities of selection by weighting each unit's study variable by the inverse of its probability of being selected, ensuring that the sample more accurately represents the finite population. Because the  second term  of \eqref{original likelihood} is irrelevant to the argument $\bZ^{\dagger}$, we concentrate on
\begin{align}
\ell\left(\bZ^{\dagger}\right) = \sum_{s=1}^{S} \sum_{(i,j)\in [\![n]\!] \times [\![m_{s}]\!]   }   \frac{r_{ij}^{(s)}}{\pi_{i}} \left\{-y_{ij}^{(s)}z_{ij}^{\dagger(s)}+g^{(s)}\left(z_{ij}^{\dagger(s)}\right)\right\}.\notag
\end{align}
Further, let $\probability(r_{ij}^{(s)}=1)=p_{ij}^{(s)}$ and  $\bP^{(s)}=(p_{ij}^{(s)})$ for $(s,i,j) \in [\![S]\!]\times  [\![n]\!] \times [\![m_{s}]\!]$. Due to complex survey sampling
and missingness, we consider the weighted loss function
\begin{align}%\label{logistic}
\ell_{w}\left(\bZ^{\dagger}\right) 
=\sum_{i=1 }^{n} \frac{1}{\pi_{i}} \left[ \sum_{s=1}^{S} \sum_{j=1}^{  m_{s}   } \frac{r_{ij}^{(s)}}{p_{ij}^{(s)}}\left\{-y_{ij}^{(s)}z_{ij}^{\dagger(s)}+g^{( s)}\left(z_{ij}^{\dagger(s)}\right)\right\}\right].\notag
\end{align}
Different from existing works \citep{alaya2019collective, robin2020main}, in this paper, we target on estimating the population-level parameters but not the sample-level ones. For example, \cite{fang2018max} considered an unweighted loss function with  max-norm  penalization  for robust estimation.  Unfortunately, under informative sampling, such unweighted loss functions may lead to biased estimation; see \cite{pfeermann1996} and \cite{mao2019matrixarxiv} for more details.

For concreteness, we focus on stratified sampling. In this scenario, we assume the existence of $H$ strata, and a total of $n$ subjects are sampled, with $n_{h}$ subjects selected from the $h$-th stratum for $h\in [\![H]\!]$. 
In addition, we assume logistic models within each stratum, 
\begin{align}\label{pij}
 p_{ij}^{(s)} = \frac{ \exp \{ (1, \bx_{i}\trans) \bzeta_{j,h}^{(s)} \}      }{ 1 +  \exp \{(1, \bx_{i}\trans) \bzeta_{j,h}^{(s)} \} },
\end{align}
for $ \sum_{l=1}^{h-1}n_{l} + 1 \leq i \leq \sum_{l=1}^{h}n_{l}$, where $\bzeta_{j,h}^{(s)} \in \mathbb{R}^{D+1}$ are the coefficients for 
the $j$-th question within the $s$-th category and the $h$-th stratum. We can fit the logistic model within each stratum and each question to estimate $\bzeta_{j,h}^{(s)}$ and further obtain $\wh{p}_{ij}^{(s)}$ in a plug-in fashion of \eqref{pij}. Therefore, the surrogate-weighted loss is 
\begin{align*}
\wh{\ell}_{w}\left(\bZ^{\dagger}\right) 
=\sum_{i=1 }^{n} \frac{1}{\pi_{i}} \left[ \sum_{s=1}^{S} \sum_{j=1}^{  m_{s}   } \frac{r_{ij}^{(s)}}{\wh{p}_{ij}^{(s)}}\left\{-y_{ij}^{(s)}z_{ij}^{\dagger(s)}+g^{( s)}\left(z_{ij}^{\dagger(s)}\right)\right\}\right].
\end{align*}
 The low-rankness of $\bZ^{\ast}$ and $\bX$ naturally allows us to formulate the estimation procedure as 
\begin{align}\label{main object}
\widehat{\bZ} =& \argmin_{\bZ^{\dagger} \in \mathbb{R}^{n \times L}} \frac{1}{NL}\wh{\ell}_{w}\left(\bZ^{\dagger}\right)  + \tau\Norm{[\bX,\bZ^{\dagger}]}_{\ast} =: \argmin_{\bZ^{\dagger} \in \mathbb{R}^{n \times L}}  \mathcal{L}_{\tau} (\bZ^{\dagger}),
\end{align}
where $\tau > 0$ is a tuning parameter. 
The concatenation $[\bX, \bZ^{\dagger}]$ helps us capture the potential nonlinear relationship between $\bZ$ and $\bX$. Similar ideas can be found in the literature of multi-task learning and multi-view data learning; see \cite{goldberg2010transduction}, \cite{zhang2012inductive}, \cite{chen2018multi} and \cite{ashraphijuo2020fundamental} for more details.
Even though we focus on stratified sampling,  the derivation of our proof can be  extended to general sampling designs. Further, in our setting, we consider different  missing mechanisms across strata. In a nutshell, our methodology is summarized in Figure~\ref{TIKZ}.

\tikzset{%
  thick arrow/.style={
     -{Triangle[angle=90:1pt 1]},
     line width=0.3cm, 
     draw=blue!20 
  },
  arrow label/.style={
    text=black,
    font=\sf,
    align=center
  },
  set mark/.style={
    insert path={
      node [sloped, midway, arrow label, node contents=#1]
    }
  }
}

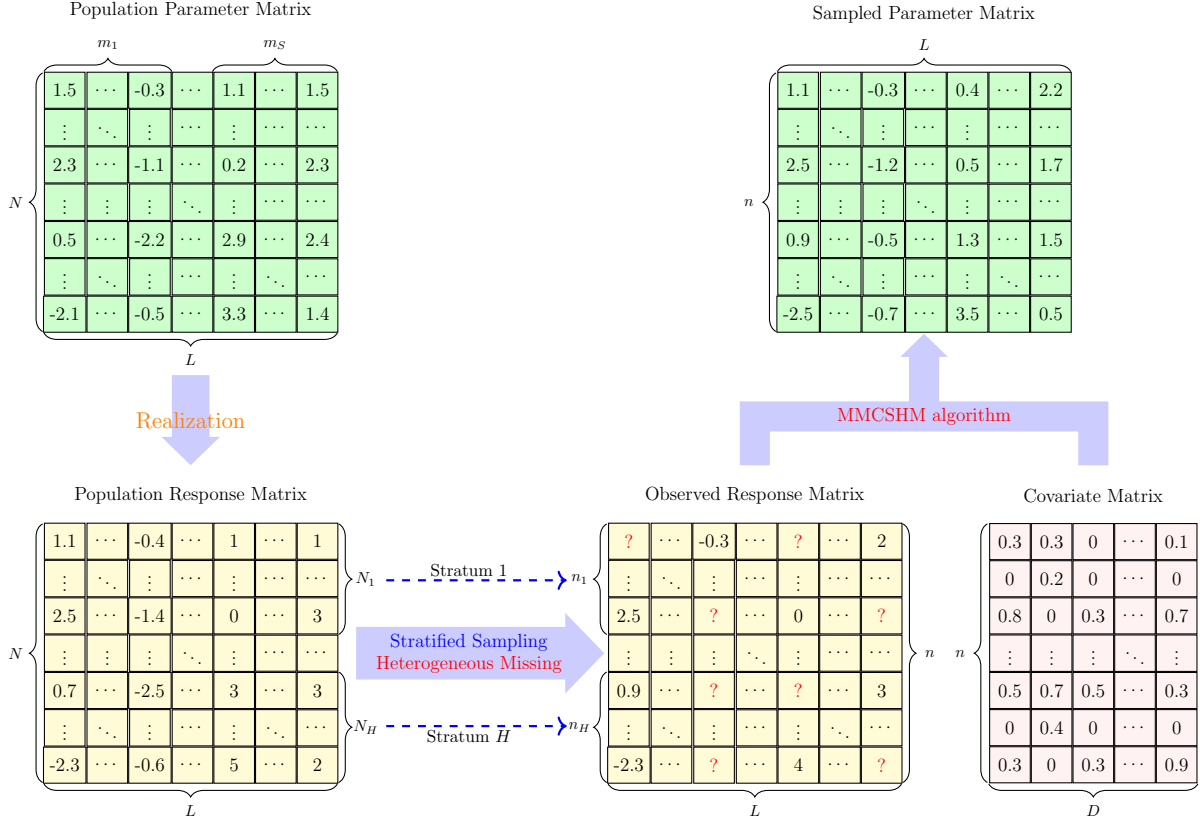
\begin{figure}[!t]
  \centering
\begin{tikzpicture}\tikzstyle{every node}=[scale=0.6]

  \matrix  [nodes = draw, nodes={fill=green!20, minimum width=9mm,
  minimum height = 8mm}] (m)
  {
    \node{1.5}; & \node{$\cdots$}; &  \node {-0.3};  &  \node{$\cdots$}; &  \node {1.1}; &  \node{$\cdots$};&  \node{1.5};  \\
    \node{$\vdots$}; &      \node{$\ddots$}; &   \node {$\vdots$}; &  \node{$\cdots$};  &       \node {$\vdots$};&  \node{$\cdots$};&  \node{$\cdots$}; \\
    \node{2.3}; &      \node{$\cdots$}; &   \node {-1.1}; &  \node{$\cdots$};  &       \node {0.2}; &  \node{$\cdots$};&  \node{2.3}; \\
    \node{$\vdots$}; &  \node{$\vdots$}; &  \node{$\vdots$}; &  \node{$\ddots$};  &      \node{$\vdots$}; &  \node{$\cdots$};&  \node{$\cdots$}; \\
    \node{0.5}; &      \node{$\cdots$}; &   \node {-2.2}; &  \node{$\cdots$};  &       \node {2.9};&  \node{$\cdots$};&  \node{$2.4$}; \\
    \node{$\vdots$}; & \node{$\ddots$}; &  \node{$\vdots$};  &  \node{$\cdots$}; &  \node {$\vdots$};&  \node{$\ddots$};&  \node{$\cdots$};   \\
    \node{-2.1}; &      \node{$\cdots$}; &   \node {-0.5}; &  \node{$\cdots$};  &       \node {3.3};&  \node{$\cdots$};&  \node{$1.4$}; \\
  };

  % \draw [dashed, red, thick] (3,-2.5) -- (3,2.5);
    \draw (0, 2.55) node {Population Parameter Matrix};
    \draw [decorate,decoration={brace,amplitude=5pt}, xshift=-0.5pt, yshift=0pt]
    (-1.95,-1.72) -- (-1.95, 1.72) node [black,midway,xshift=-0.6cm] 
    {\footnotesize $N$};
        \draw [decorate,decoration={brace,amplitude=5pt,mirror}, xshift=0pt, yshift=-0.5pt]
    (-1.95,-1.72) -- (1.95, -1.72) node [black,midway,yshift=-0.6cm] 
    {\footnotesize $L$};
    \draw [decorate,decoration={brace,amplitude=5pt}, xshift=0pt, yshift=  0.5pt]
    (-1.95, 1.72) -- (-0.25, 1.72) node [black,midway,yshift=0.6cm] 
    {\footnotesize $m_{1}$};
    \draw [decorate,decoration={brace,amplitude=5pt}, xshift=0pt, yshift=  0.5pt]
    (0.29, 1.72) -- (1.95, 1.72) node [black,midway,yshift=0.6cm] 
    {\footnotesize $m_{S}$};

    \begin{scope}[xshift= 9.75cm]
 \matrix  [nodes = draw, nodes={fill=green!20,minimum width=9mm,
 minimum height = 8mm}] 
 {
  \node{1.1}; & \node{$\cdots$}; &  \node {-0.3};  &  \node{$\cdots$}; &  \node {0.4}; &  \node{$\cdots$};&  \node{2.2};  \\
  \node{$\vdots$}; &      \node{$\ddots$}; &   \node {$\vdots$}; &  \node{$\cdots$};  &       \node {$\vdots$};&  \node{$\cdots$};&  \node{$\cdots$}; \\
  \node{2.5}; &      \node{$\cdots$}; &   \node {-1.2}; &  \node{$\cdots$};  &       \node {0.5}; &  \node{$\cdots$};&  \node{1.7}; \\
  \node{$\vdots$}; &  \node{$\vdots$}; &  \node{$\vdots$}; &  \node{$\ddots$};  &      \node{$\vdots$}; &  \node{$\cdots$};&  \node{$\cdots$}; \\
  \node{0.9}; &      \node{$\cdots$}; &   \node {-0.5}; &  \node{$\cdots$};  &       \node {1.3};&  \node{$\cdots$};&  \node{$1.5$}; \\
  \node{$\vdots$}; & \node{$\ddots$}; &  \node{$\vdots$};  &  \node{$\cdots$}; &  \node {$\vdots$};&  \node{$\ddots$};&  \node{$\cdots$};   \\
  \node{-2.5}; &      \node{$\cdots$}; &   \node {-0.7}; &  \node{$\cdots$};  &       \node {3.5};&  \node{$\cdots$};&  \node{$0.5$}; \\
};

    % \draw (0,1.75) node {Gaussian};
    \draw [decorate,decoration={brace,amplitude=5pt}, xshift=-0.5pt, yshift=0pt]
    (-1.95,-1.72) -- (-1.95, 1.72) node [black,midway,xshift=-0.6cm] 
    {\footnotesize $n$};
        \draw [decorate,decoration={brace,amplitude=5pt}, xshift=0pt, yshift=  0.5pt]
    (-1.95, 1.72) -- (1.95, 1.72) node [black,midway,yshift=0.6cm] 
    {\footnotesize $L$};
    \draw (0,2.5) node { Sampled Parameter Matrix};
  \end{scope}

\draw [thick arrow, line width=5mm] (0, -2.3)  -> (0, -3.5);
\node at  (0,-2.9) {\textcolor{orange}{\large{Realization}}};

  % Upper Part
  \begin{scope}[ yshift = -6cm]

  %\shadedraw[inner color=red,outer color=orange,draw=orange, rounded corners] (1.65,-2.7) rectangle +(10,5.4);

  \matrix  [nodes = draw, nodes={fill=yellow!20, minimum width=9mm,
  minimum height = 8mm}] at (0, 0)
  {
    \node{1.1}; & \node{$\cdots$}; &  \node {-0.4};  &  \node{$\cdots$}; &  \node {1}; &  \node{$\cdots$};&  \node{1};  \\
    \node{$\vdots$}; &      \node{$\ddots$}; &   \node {$\vdots$}; &  \node{$\cdots$};  &       \node {$\vdots$};&  \node{$\cdots$};&  \node{$\cdots$}; \\
    \node{2.5}; &      \node{$\cdots$}; &   \node {-1.4}; &  \node{$\cdots$};  &       \node {0}; &  \node{$\cdots$};&  \node{3}; \\
    \node{$\vdots$}; &  \node{$\vdots$}; &  \node{$\vdots$}; &  \node{$\ddots$};  &      \node{$\vdots$}; &  \node{$\cdots$};&  \node{$\cdots$}; \\
    \node{0.7}; &      \node{$\cdots$}; &   \node {-2.5}; &  \node{$\cdots$};  &       \node {3};&  \node{$\cdots$};&  \node{$3$}; \\
    \node{$\vdots$}; & \node{$\ddots$}; &  \node{$\vdots$};  &  \node{$\cdots$}; &  \node {$\vdots$};&  \node{$\ddots$};&  \node{$\cdots$};   \\
    \node{-2.3}; &      \node{$\cdots$}; &   \node {-0.6}; &  \node{$\cdots$};  &       \node {5};&  \node{$\cdots$};&  \node{$2$}; \\
  };
  % \draw [dashed, red, thick] (3,-2.5) -- (3,2.5);
  \draw (0, 2.1) node {Population Response Matrix};
  \draw [decorate,decoration={brace,amplitude=5pt}, xshift=-0.5pt, yshift=0pt]
  (-1.95,-1.72) -- (-1.95, 1.72) node [black,midway,xshift=-0.6cm] 
  {\footnotesize $N$};
      \draw [decorate,decoration={brace,amplitude=5pt,mirror}, xshift=0pt, yshift=-0.5pt]
  (-1.95,-1.72) -- (1.95, -1.72) node [black,midway,yshift=-0.6cm] 
  {\footnotesize $L$};
  \draw [decorate,decoration={brace,amplitude=5pt,mirror}, xshift=0.5pt, yshift=0pt]
  (1.95, 0.25) -- (1.95,  1.72) node [black,midway,xshift=0.6cm] 
  {\footnotesize $N_{1}$};
  \draw [decorate,decoration={brace,amplitude=5pt,mirror}, xshift=0.5pt, yshift=0pt]
  (1.95, -1.72) -- (1.95,  -0.25) node [black,midway,xshift=0.6cm] 
  {\footnotesize $N_{H}$};

  \draw [thick arrow, line width=7mm] (2.2, 0)  -> (5.5, 0);
  % \node at  (3.3,0) {\textcolor{orange}{\large{Stratum $H$}}}; 
  \node at (3.7, 0.15){ \textcolor{blue}{Stratified Sampling}};
  \node at (3.7, -0.15){\textcolor{red}{Heterogeneous Missing}};

  \draw [->,  dashed, line width=0.3mm, blue] (2.6,0.97) --  (5,0.97) ;
  \draw [->,  dashed, line width=0.3mm, blue] (2.6,-0.97) --  (5,-0.97) ;
  \node at (3.7, 1.1){ \small{Stratum 1}};
  \node at (3.7, -1.1){ \small{Stratum $H$}};

    \end{scope}
  \begin{scope}[xshift = 7.5 cm , yshift = - 6cm]
 \matrix  [nodes = draw, nodes={fill=yellow!20,minimum width=9mm,
 minimum height = 8mm}] 
 {
  \node{\textcolor{red}{?}}; & \node{$\cdots$}; &  \node {-0.3};  &  \node{$\cdots$}; &  \node {\textcolor{red}{?}}; &  \node{$\cdots$};&  \node{2};  \\
  \node{$\vdots$}; &      \node{$\ddots$}; &   \node {$\vdots$}; &  \node{$\cdots$};  &       \node {$\vdots$};&  \node{$\cdots$};&  \node{$\cdots$}; \\
  \node{2.5}; &      \node{$\cdots$}; &   \node {\textcolor{red}{?}}; &  \node{$\cdots$};  &       \node {0}; &  \node{$\cdots$};&  \node{\textcolor{red}{?}}; \\
  \node{$\vdots$}; &  \node{$\vdots$}; &  \node{$\vdots$}; &  \node{$\ddots$};  &      \node{$\vdots$}; &  \node{$\cdots$};&  \node{$\cdots$}; \\
  \node{0.9}; &      \node{$\cdots$}; &   \node {\textcolor{red}{?}}; &  \node{$\cdots$};  &       \node {\textcolor{red}{?}};&  \node{$\cdots$};&  \node{$3$}; \\
  \node{$\vdots$}; & \node{$\ddots$}; &  \node{$\vdots$};  &  \node{$\cdots$}; &  \node {$\vdots$};&  \node{$\ddots$};&  \node{$\cdots$};   \\
  \node{-2.3}; &      \node{$\cdots$}; &   \node {\textcolor{red}{?}}; &  \node{$\cdots$};  &       \node {4};&  \node{$\cdots$};&  \node{\textcolor{red}{?}}; \\
};
% \draw [dashed, red, thick] (3,-2.5) -- (3,2.5);
\draw (0, 2.1) node {Observed Response Matrix};

\draw [decorate,decoration={brace,amplitude=5pt}, xshift=0.5pt, yshift=0pt]
(1.95,1.72) -- (1.95, -1.72) node [black,midway,xshift=0.6cm] 
{\footnotesize $n$};
    \draw [decorate,decoration={brace,amplitude=5pt,mirror}, xshift=0pt, yshift=-0.5pt]
(-1.95,-1.72) -- (1.95, -1.72) node [black,midway,yshift=-0.6cm] 
{\footnotesize $L$};
\draw [decorate,decoration={brace,amplitude=5pt,mirror}, xshift=-0.5pt, yshift=0pt]
(-1.95, 1.72) -- (-1.95,  0.25) node [black,midway,xshift=-0.6cm] 
{\footnotesize $n_{1}$};
\draw [decorate,decoration={brace,amplitude=5pt,mirror}, xshift=-0.5pt, yshift=0pt]
(-1.95, -0.25) -- (-1.95, -1.72 ) node [black,midway,xshift=-0.6cm] 
{\footnotesize $n_{H}$};

\draw [ thick, line width=4mm, blue!20 ] (-0,2.5) -- (-0,3.15)  --  (4.5,3.15) --  (4.5,2.5) ;
\draw [thick arrow, line width=4mm] (2.25, 3.25)  -> (2.25, 4.25);
\draw (2.25, 3.15) node[red] {MMCSHM algorithm};

  \end{scope}

   \begin{scope}[xshift=12cm  , yshift = -6cm]
 \matrix  [nodes = draw, nodes={fill=pink!20,minimum width=9mm,
 minimum height = 8mm}] 
  {
    \node{0.3}; & \node{0.3}; &  \node {0};  &  \node{$\cdots$}; &  \node {0.1};   \\
    \node{0}; &      \node{0.2}; &   \node {0}; &  \node{$\cdots$};  &       \node {0}; \\
    \node{0.8}; &      \node{0}; &   \node {0.3}; &  \node{$\cdots$};  &       \node {0.7}; \\
    \node{$\vdots$}; &  \node{$\vdots$}; &  \node{$\vdots$}; &  \node{$\ddots$};  &      \node{$\vdots$}; \\
    \node{0.5}; &      \node{0.7}; &   \node {0.5}; &  \node{$\cdots$};  &       \node {0.3}; \\
    \node{0}; &      \node{0.4}; &   \node {0}; &  \node{$\cdots$};  &       \node {0}; \\
    \node{0.3}; &      \node{0}; &   \node {0.3}; &  \node{$\cdots$};  &       \node {0.9}; \\
  };
    \draw [decorate,decoration={brace,amplitude=5pt,mirror}, xshift=0pt, yshift=-0.5pt]
    (-1.39,-1.72) -- (1.39,-1.72) node [black,midway,yshift=-0.6cm] 
{\footnotesize $D$};
\draw [decorate,decoration={brace,amplitude=5pt}, xshift=-0.5pt, yshift=0pt]
(-1.39,-1.72) -- (-1.39, 1.72) node [black,midway,xshift=-0.6cm] 
{\footnotesize $n$};
    \draw (0,2.1) node {Covariate Matrix};
  \end{scope}

\end{tikzpicture}
%\vspace{.3in}
\caption{
Algorithm illustration, where a question mark represents a missing value. }\label{TIKZ}
\end{figure}

\subsection{Estimation Algorithm}

We propose a mixed matrix completion method in survey with heterogeneous missing  (MMCSHM)  algorithm to obtain the estimator in \eqref{main object} summarized in Algorithm~\ref{algorithm}. 
For ease of presentation, denote
\begin{align}
 \nabla \wh{\ell}_{w}(\bZ) = \sum_{i=1 }^{n} \frac{1}{\pi_{i}} \left[ \sum_{s=1}^{S} \sum_{j=1}^{  m_{s}   } \frac{r_{ij}^{(s)}}{\wh{p}_{ij}^{(s)}}\left\{-y_{ij}^{(s)}+(g^{( s)})'\left(z_{ij}^{(s)}\right)\right\}\right]\bE_{ij}^{(s)},
\end{align}
where $\bE_{ij}^{(s)}$ is the indicator matrix for the $i$-th subject and the $j$-th  variable within the $s$-th category.
Furthermore, for a matrix $\bS \in \mathbb{\bR}^{n\times L}$, a singular value decomposition is $\bS=\bU\bW\bV\trans$, where $\bU\in \mathbb{R}^{n\times r}$, 
$\bV \in \mathbb{R}^{L \times r}$, $\bW = \mbox{diag}(\sigma_{1}, \ldots, \sigma_{r})$ and $\sigma_1 \geq \ldots \geq \sigma_{r} > 0.$ 
Let $\bW_{1} = \mbox{diag}(\sigma_{1}, 0, \ldots, 0)$. Denote the rank-1 approximation of $\bS$ as $\mbox{rank-1 SVD}(\bS) = \bU\bW_{1}\bV\trans$, and it is used as the initialization for our proposed estimation procedure.  Define the singular value thresholding operator \citep{cai2010singular}  for $\bS$ with the parameter $\tau$ as 
$
   \mbox{SVT}_{\tau}(\bS) = \bU  \mbox{diag}( (\sigma_{1} -\tau   )_{+} \ldots (\sigma_{r} -\tau   )_{+}  )    \bV\trans,
$
where $x_{+} = \max\{x, 0\}$ for $x \in \mathbb{R}$.
Specifically, in the first stage, we opt for logistic regression to model the entry-wise missing mechanism with auxiliary information. In the second stage (Step 2 to Step 12 of Algorithm~\ref{algorithm}), with the fitted response probability,  we leverage a decent version of FISTA \citep{beck2017first} to solve \eqref{main object}. 
 The computational complexity for a single iteration is mainly determined by the singular value decomposition, which is $\mathcal{O}(n(D+L)\min\{n, D+L\})$. Consequently, for $K$ iterations, the overall computational complexity becomes $\mathcal{O}(n(D+L)K\min\{n, D+L\})$.
The convergence analysis of the proposed algorithm will be presented in the following section.

\begin{algorithm}[!t]
   \caption{MMCSHM algorithm}
   \label{algorithm}
   \KwIn{The observed matrix $\bY$, missingness matrix $\bR$, covariate matrix $\bX$, sample size $\{n_{h}\}_{h=1}^{H}$, 
   population size $N$, sampling probability $\{\pi_{i}\}_{i=1}^{n}$, tuning parameter $\tau$, and learning depth $K$.}
	\kwInit{$[\bU_0, \bW_{0}, \bV_0] = \mbox{rank}-1\ \mbox{SVD}(\bR \circ \bY) $. $\bZ_{1}^{(0)} = \bZ_{2}^{(0)} = \bU_0 \bW_{0} \bV_0\trans$. }
   Compute $\{\wh{p}_{ij}^{(s)}\}$ by fitting the logistic model with $\bX$ and $\bR$ within the corresponding strata.\\
   \For {$k = 1$ \KwTo $K$, }
   	{ 
    	Compute $\theta_{k} = 2/(k+1)$.\\
    	Compute $\bQ = (1-\theta_{k})\bZ_{1}^{(k-1)} +  \theta_{k} \bZ_{2}^{(k-1)}. $\\
    	Compute $\bT = \bQ - \tau^{-1} \mbox{SVT}_{\tau}( \nabla\wh{\ell}_{w}(\bQ)  )$. \\
    Compute $\wt{\bT} = \mbox{SVT}_{\tau} [\bX, \bT]$.\\
    	Compute $\wt{\bZ}_{1}^{(k)} = \wt{\bT}[ \bzero_{n\times D}  , \bI_{n \times L} ]\trans$. \\
   		 \uIf{ $\mathcal{L}_{\tau}(\wt{\bZ}_{1}^{(k)}) < \mathcal{L}_{\tau}(\bZ_{1}^{(k)})$ }{
    	$ \bZ_{1}^{(k)} = \wt{\bZ}_{1}^{(k)} $;
 		}
 		\Else{
    		$ \bZ_{1}^{(k)} =  \bZ_{1}^{(k-1)}$; 
  		}
      $\bZ_{2}^{(k)} = \bZ_{1}^{(k-1)}+  \theta_{k}^{-1} (  \wt{\bZ}_{1}^{(k)}  - \bZ_{1}^{(k-1)}  ). $\\
   	}
   \KwOut{$\bZ_{1}^{(K)}$.}
\end{algorithm}

\section{Theoretical Guarantee}

In this section, we present the statistical guarantee and the convergence analysis of the proposed algorithm. Before we dive into the theoretical results, we introduce the following technical assumptions.

\begin{assumption}\label{bound of survey}
 There exist   positive constants $\alpha_{L}$ and $\alpha_{U}$ such that
$
    \alpha_{L} \leq (N\pi_{i})^{-1}  n_{h}         \leq \alpha_{U},
$
for $ \sum_{l=1}^{h-1}n_{l} + 1 \leq i \leq \sum_{l=1}^{h}n_{l}$ and $h\in [\![H]\!]$.
\end{assumption}

\begin{assumption}\label{bound of missing rate}
 There exists a positive constant $p_{\rm min} > 0$ such that $p_{\rm min} \leq p_{ij}^{(s)}$ 
for $(s,i,j) \in [\![S]\!]\times  [\![n]\!] \times [\![m_{s}]\!]$.
\end{assumption}

\begin{assumption}\label{bound of Z}
 There exists a positive constant $\beta$, such that
$
  \Norm{\bZ^{\ast}}_{\infty}  \leq \beta.
$
\end{assumption}

\begin{assumption}\label{bound of exponential family}
   Denote 
   $
   C( z_{ij}^{\ast(s)}   ) = \inf_{0\leq r \leq r(\bZ)}r^{-1}\mathbb{E}_{z_{ij}^{\ast(s)} }[ \exp\{   r|Y_{ij}^{(s)} - (g^{(s)})'(z_{ij}^{\ast (s)}) |  \}], 
      $
 where $r(\bZ)$ is the natural parameter space for the exponential family. Assume that there exists a constant $C_{\bZ}$ such that $ C( z_{ij}^{\ast(s)}   ) \leq C_{\bZ} < \infty$ for $(s,i,j) \in [\![S]\!]\times  [\![n]\!] \times [\![m_{s}]\!]$.
   \end{assumption}

\begin{assumption}\label{bound of second derivative of g}
Let $\mathcal{D} = [ -\beta - \epsilon, \beta + \epsilon]$, for some $\epsilon > 0$. 
For any $z \in \mathcal{D}$, there exist positive constants $L_{\beta}$ and $U_{\beta}$, such that
$L_{\beta}  \leq   (g^{(s)})''(z) \leq  U_{\beta},$
where $g^{(s)}$ is the link function of the exponential family for the $s$-th category of the questionaire, for $s \in [\![S]\!] $.
\end{assumption}

Assumption~\ref{bound of survey} is commonly used to control the sampling weights in survey sampling; see Section 1.3 in \cite{wayne2009sampling}. Assumption~\ref{bound of missing rate} controls the response probabilities for the entries. Assumption~\ref{bound of Z} implies that the parameters are bounded. Assumption~\ref{bound of exponential family}  requires all $ C( z_{ij}^{\ast(s)}   )$'s  bounded from above by a constant $C_{\bZ}$. Under the framework of the canonical exponential family, Assumption~\ref{bound of second derivative of g} indicates that $\mbox{Var}(Y_{ij}^{(s)}|z_{ij}^{\ast (s)}) =  (g^{(s)})''(z_{ij}^{\ast (s)}) > 0$ with extended support. 

Denoting $\bP^{(s)\dagger} = ( (p_{ij}^{(s)})^{-1} )_{ij}$,  the inverse probability matrix is $\bP^{\dagger} = [\bP^{(1)\dagger}, \ldots, \bP^{(S)\dagger}]$. Similarly, 
let $\hat{\bP}^{\dagger} = [\hat{\bP}^{(1)\dagger}, \ldots,  \hat{\bP}^{(S)\dagger}]$, where $\hat{\bP}^{(s)\dagger} = ((\hat{p}_{ij}^{(s)})^{-1})_{ij}$. Further, denote
\begin{equation*}
   \zeta_{\bR}  = \max\left\{ \Norm{\bR \circ \left( \wh{\bP}^{\dagger}  - \bP^{\dagger} \right)  }_{\infty, 2}       , 
      \Norm{\bR\trans \circ ( \wh{\bP}^{\dagger}  - \bP^{\dagger} )\trans  }_{\infty, 2}   \right\},
   \end{equation*}
   and
   \begin{align*}
  & \wt{\Delta} = \max \left\{  \frac{\left\{C_{\bZ} (n\vee L)\log(n+L) \right\}^{1/2} } {\sqrt{p_{\rm min}} nL  }, 
    \frac{ \alpha_{U} C_{\bZ} \zeta_{\bR} (\log(n+L))^{1+\delta}   }{   nL } ,\right. \notag \\
   & \left. \quad\quad\quad\quad\quad\quad\quad\quad\quad\quad\quad  \sqrt{\frac{  \alpha_{U}\beta L_{\beta}^{2} |1/p_{\rm min}|(n \vee L)\log(n+L)   }{nL}}\right\}
\end{align*} 
for some $\delta > 0$.   The upper bound of $\norm{\hat{\bZ} - \bZ^{\ast}}_{F}^{2}$ can be obtained by the following theorem,
which provides a statistical guarantee of our method.

\begin{comment}
The term $T1$ can be handled by a general matrix norm inequality. The rest terms can be handled by matrix trace inequality, in addition with the technique in Alaya et. al (2020) to handle $\norm{\bG}$ and the technique in Mao et al. (2019) to handle $\wh{\bP}^{\dagger} - \bP^{\dagger}$ and $\bR \circ \bP^{\dagger} - \bJ$. The matrix $\bPi^{\dagger}$ can be handled element-wisely.  
\end{comment}

\begin{theorem}\label{statistical guarantee}
    Suppose Assumptions 1$\sim$5 hold. Then there exist positive constants $C_{1}, C_{2}, C_{3}$ such that for $\tau \geq C_{3} \wt{\Delta}$, 
    \begin{align*}
    \frac{1}{nL}\Norm{\wh{\bZ} - \bZ^{\ast}}_{F}^{2} \leq & 
    \max\left\{
     \frac{ C_{1} \alpha_{L}}  { {\alpha_{U} \beta^{2}} } \sqrt{\frac{\log(n+L)}{  nL } }, 
    \frac{ C_{2}   {\rm rank} (\bM^{\ast} )} {  \alpha_{L}p_{\rm min} } \left\{    \tau^{2} + \frac{  \alpha_{U}^{2}  \log(n+L)   }{n  \wedge L}  \right\}
    \right\}
    \end{align*}
    holds with probability at least $1 - 3/(n+L)$, where $\bM^{\ast} = [\bX, \bZ^{\ast}]$.
\end{theorem}

\begin{proof}
First of all, by $R^{\ast}(\wh{\bZ}) \leq R^{\ast}(\bZ)$, we have
\begin{align*}
\frac{1}{NL}  \wh{\ell}_{w}(\wh{\bZ}) + \tau \norm{ [\bX, \wh{\bZ} ]   }_{\ast}
\leq \frac{1}{NL}  \wh{\ell}_{w}(\bZ^{\ast}) + \tau \Norm{ [\bX, \bZ^{\ast} ]   }_{\ast},\notag
\end{align*}
which yields
\begin{align}
\frac{1}{NL}\sum_{i=1}^{n}\left\{   \wh{\ell}_{w, i}(\wh{\bZ}) -  \wh{\ell}_{w, i}(\bZ^{\ast})   \right\}
\leq \tau \left(  \norm{ [\bX, \bZ^{\ast} ]   }_{\ast} - \norm{ [\bX, \wh{\bZ} ]   }_{\ast}       \right).\notag
\end{align}
Expand the terms $\ell_{w, i}(\wh{\bZ})$ and $\ell_{w, i}(\bZ^{\ast})$, we obtain
\begin{align*}
&\frac{1}{nL}\sum_{i=1}^{n}\frac{n}{N\pi_{i}} \sum_{s=1}^{S}\sum_{j=1}^{m_{s}} \frac{r_{ij}^{(s)}}{\wh{p}_{ij}^{(s)}}
\left\{  g^{(s)}(\wh{z}_{ij}^{(s)}) - g^{(s)}(z_{ij}^{\ast(s)})       \right\}\notag\\
\leq& \tau \left(  \norm{ [\bX, \bZ^{\ast} ]   }_{\ast} - \norm{ [\bX, \wh{\bZ} ]   }_{\ast}       \right)
+ \frac{1}{nL}\sum_{i=1}^{n}\frac{n}{N\pi_{i}} \sum_{s=1}^{S}\sum_{j=1}^{m_{s}} \frac{r_{ij}^{(s)}}{\wh{p}_{ij}^{(s)}}
y_{ij}^{(s)} \left( \wh{z}_{ij}^{(s)} - z_{ij}^{\ast(s)}       \right). 
\end{align*}
By definition of Bregman divergence, it follows that
\begin{align}\label{main basic}
&\frac{1}{nL}\sum_{i=1}^{n}\frac{n}{N\pi_{i}} \sum_{s=1}^{S} \sum_{j=1}^{m_{s}} \frac{r_{ij}^{(s)}}{\wh{p}_{ij}^{(s)}}d_{g^{(s)}}(\wh{z}_{ij}^{(s)}, z_{ij}^{\ast(s)})\notag\\
\leq&  \tau \left(  \Norm{ [\bX, \bZ^{\ast} ]   }_{\ast} - \norm{ [\bX, \wh{\bZ} ]   }_{\ast}       \right) \notag\\
&\quad\quad-
\frac{1}{nL}\sum_{i=1}^{n}\frac{n}{N\pi_{i}} \sum_{s=1}^{S}\sum_{j=1}^{m_{s}} \frac{r_{ij}^{(s)}}{\wh{p}_{ij}^{(s)}}
\left\{  \left(g^{(s)}\right)'(z_{ij}^{\ast(s)})     - y_{ij}^{(s)} \right\} \left( \wh{z}_{ij}^{(s)} - z_{ij}^{\ast(s)}       \right) .
\end{align}
To handle the left-hand-side of \eqref{main basic}, by   $L_{\beta}(x-y)^{2} \leq 2d_{g(s)}(x-y) \leq U_{\beta}(x-y)^{2}$, we have 
\begin{align}\label{bregman plug in}
&\frac{1}{nL}\sum_{i=1}^{n} \frac{n}{N\pi_{i}}\sum_{s=1}^{S}\sum_{j=1}^{m_{s}}\frac{r_{ij}^{(s)}}{\wh{p}_{ij}^{(s)}} \left( \wh{z}_{ij}^{(s)} - z_{ij}^{\ast(s)}   \right)^{2}
\leq\frac{2}{nLL_{\beta}}\sum_{i=1}^{n} \frac{n}{N\pi_{i}}\sum_{s=1}^{S}\sum_{j=1}^{m_{s}}\frac{r_{ij}^{(s)}}{\wh{p}_{ij}^{(s)}}d_{g^{(s)}}( \wh{z}_{ij}^{(s)}, z_{ij}^{\ast(s)}   ).
\end{align}
Plug \eqref{bregman plug in} into \eqref{main basic}, it turns out that
\begin{align}\label{focus}
&\frac{L_{\beta}}{ 2nL  }  \sum_{i=1}^{n} \frac{n}{N\pi_{i}}\sum_{s=1}^{S}\sum_{j=1}^{m_{s}}\frac{r_{ij}^{(s)}}{\wh{p}_{ij}^{(s)}} \left( \wh{z}_{ij}^{(s)} - z_{ij}^{\ast(s)}   \right)^{2}\notag\\
\leq & \tau \left(  \Norm{ [\bX, \bZ^{\ast} ]   }_{\ast} - \norm{ [\bX, \wh{\bZ} ]   }_{\ast}       \right)\notag\\
&\quad\quad-
\frac{1}{nL}\sum_{i=1}^{n}\frac{n}{N\pi_{i}} \sum_{s=1}^{S}\sum_{j=1}^{m_{s}} \frac{r_{ij}^{(s)}}{\wh{p}_{ij}^{(s)}}
\left\{  \left(g^{(s)}\right)'(z_{ij}^{\ast(s)})     - y_{ij}^{(s)} \right\} \left( \wh{z}_{ij}^{(s)} - z_{ij}^{\ast(s)}       \right). 
\end{align}
Recall that the inverse probability matrix $\bP^{\dagger} = [\bP^{(1)\dagger}, \ldots, \bP^{(S)\dagger}]$ and its estimated surrogate matrix $\wh{\bP}^{\dagger} = [\wh{\bP}^{(1)\dagger}, \ldots,  \wh{\bP}^{(S)\dagger}]$, where $\bP^{(s)\dagger} = ( (p_{ij}^{(s)})^{-1} )_{ij}$ and $\wh{\bP}^{(s)\dagger} = ((\wh{p}_{ij}^{(s)})^{-1})_{ij}$.  Let $\mathbf{1}_{m}$ be a vector of all ones elements with length $m$. Denote $\bPi^{\dagger} = N^{-1}(n_{h(1)}\pi_{1}^{-1}, \ldots, n_{h(n)}\pi_{n}^{-1})\trans\mathbf{1}_{L}\trans $. Furthermore, let $\bJ = \mathbf{1}_{n}\mathbf{1}_{L}\trans$. Inequality \eqref{focus} can be written into the following matrix representation
\begin{align}\label{matrix first}
&\frac{L_{\beta}}{2nL} \langle \bR \circ \wh{\bP}^{\dagger}, \bPi^{\dagger} \circ (\wh{\bZ} - \bZ^{\ast})\circ  (\wh{\bZ} - \bZ^{\ast})   \rangle \notag\\
\leq &\tau \left(  \Norm{ [\bX, \bZ^{\ast} ]   }_{\ast} - \norm{ [\bX, \wh{\bZ} ]   }_{\ast}       \right)+
\frac{1}{nL}  \langle \bR \circ \wh{\bP}^{\dagger} \circ \bG,  \bPi^{\dagger} \circ (\wh{\bZ} - \bZ^{\ast})\rangle,
\end{align}
where $\circ$ is the Hadamard product between matrices,  $\bG = [\bG^{(1)}, \ldots, \bG^{(S)}]$ and $\bG^{(s)} = (    y_{ij}^{(s)}  - (g^{(s)})'(z_{ij}^{\ast (s)})    )_{ij}$. We further expand the terms in \eqref{matrix first} as
\begin{align}\label{decomposition one}
 &\langle   \bR \circ \wh{\bP}^{\dagger}, \bPi^{\dagger} \circ (\wh{\bZ} - \bZ^{\ast})\circ  (\wh{\bZ} - \bZ^{\ast})   \rangle \notag\\
 = &  \langle   \bR \circ \bP^{\dagger},  \bPi^{\dagger} \circ (\wh{\bZ} - \bZ^{\ast})\circ  (\wh{\bZ} - \bZ^{\ast})    \rangle  + \langle    \bR \circ \wh{\bP}^{\dagger}  -   \bR \circ \bP^{\dagger}, \bPi^{\dagger} \circ (\wh{\bZ} - \bZ^{\ast})\circ  (\wh{\bZ} - \bZ^{\ast})    \rangle.
\end{align}
Therefore, plug $\eqref{decomposition one}$  into \eqref{matrix first}, we have
\begin{align}
&  \frac{L_{\beta}}{  2nL  }\langle   \bR \circ \bP^{\dagger},  \bPi^{\dagger} \circ (\wh{\bZ} - \bZ^{\ast})\circ  (\wh{\bZ} - \bZ^{\ast})    \rangle  \notag\\
\leq &
\underbrace{ \tau \left(  \Norm{ [\bX, \bZ^{\ast} ]   }_{\ast} - \norm{ [\bX, \wh{\bZ} ]   }_{\ast}       \right)}_{T1} + \underbrace{ \frac{1}{nL} \langle  \bR \circ \wh{\bP}^{\dagger}  \circ \bG, \bPi^{\dagger} (\wh{\bZ} - \bZ^{\ast})\rangle}_{T2}\notag\\
&- \underbrace{ \frac{L_{\beta}}{2nL} \langle   \bR \circ (\wh{\bP}^{\dagger} - \bP^{\dagger})  ,  \bPi^{\dagger} \{(\wh{\bZ} - \bZ^{\ast})\circ  (\wh{\bZ} - \bZ^{\ast})  \} \rangle}_{T3}.\label{main decomposition}
\end{align}
Let $\tilde{\Delta} =  2 \Delta^{(1)}/L_{\beta} + \Delta^{(2)}  $ and $\tilde{\tau} = 2\tau/L_{\beta}$. By \eqref{main decomposition} and Lemma S1-S3, with probability at least $1- 3/(n+L)$, there exists a constant $\tilde{C} > 0$ such that
\begin{align*}
\frac{1}{nL} \langle   \bR \circ \bP^{\dagger},  \bPi^{\dagger} \circ (\wh{\bZ} - \bZ^{\ast})\circ  (\wh{\bZ} - \bZ^{\ast})    \rangle  
 \leq 
 \tilde{\tau} \left(  \Norm{ [\bX, \bZ^{\ast} ]   }_{\ast} - \norm{ [\bX, \wh{\bZ} ]   }_{\ast}       \right) + \tilde{C} \tilde{\Delta}\Norm{ \bZ^{\ast}  - \hat{\bZ}  }_{\ast}.
\end{align*}
Denote $\bM^{\ast} = [\bX, \bZ^{\ast}]$. Let the singular value decomposition of $\bM^{\ast}= \sum_{i=1}^{r_{\bM^{\ast}}   }  \sigma_{i, \bM^{\ast}}   \bu_{i, \bM^{\ast}}\bv_{i, \bM^{\ast}}\trans $. Denote $\bA_{\bu} = [ \bu_{1, \bM^{\ast}}, \ldots, \bu_{r_{\bM^{\ast}} , \bM^{\ast}}  ]$ and $\bA_{\bv} = [ \bv_{1, \bM^{\ast}}, \ldots, \bv_{r_{\bM^{\ast}} , \bM^{\ast}}  ]$ . The projection operators can be defined as
\begin{align}\label{perpendicular}
&\mathcal{P}_{\bA^{\perp}} (\bB) = \bP_{\bA_{\bu}^{\perp}}\bB \bP_{\bA_{\bv}^{\perp}}, \quad\mathcal{P}_{\bA} (\bB) = \bB -  \mathcal{P}_{\bA^{\perp}} (\bB) = \bP_{A_{\bu}}\bB - \bP_{\bA_{\bu}^{\perp}}\bB\bP_{\bA_{\bv}},
\end{align}
where $ \bP_{\bA_{\bu}^{\perp}}$ is the projection matrix generated by $\bA_{\bu}^{\perp}$ and 
$\bA_{\bu}^{\perp} = \bA_{\bu} - \bP_{\bA_{\bu}}(\bA_{\bu})$.
  It  follows that
\begin{align}
 &\Norm{ [\bX, \bZ^{\ast} ]   }_{\ast} - \norm{ [\bX, \wh{\bZ} ] }_{\ast}\notag\\
 \leq& \Norm{ \mathcal{P}_{  \bM^{\ast} } \left( [\bX, \bZ^{\ast} ]  - [\bX, \wh{\bZ} ]    \right) }_{\ast}- \Norm{ \mathcal{P}_{  (\bM^{\ast})^{\perp} } \left( [\bX, \bZ^{\ast} ]  - [\bX, \wh{\bZ} ]    \right) }_{\ast} \notag\\
 \leq& \Norm{ \mathcal{P}_{  \bM^{\ast} } \left( [\bzero, \bZ^{\ast} -  \wh{\bZ} ]   \right) }_{\ast} - \Norm{ \mathcal{P}_{  (\bM^{\ast})^{\perp} } \left( [\bzero, \bZ^{\ast} -  \wh{\bZ} ]     \right) }_{\ast}. \label{nuclear norm difference}
\end{align}
Further, 
\begin{align*}
\Norm{ \mathcal{P}_{  \bM^{\ast} } \left( [\bzero, \bZ^{\ast} -  \wh{\bZ} ]   \right) }_{\ast}
\leq  &\sqrt{ {\rm rank}( \mathcal{P}_{  \bM^{\ast}   } ) }  \Norm{  [\bzero, \bZ^{\ast} -  \wh{\bZ} ]  }_{F}
= \sqrt{ {\rm rank}(  \mathcal{P}_{  \bM^{\ast}   }) }  \Norm{  \bZ^{\ast} -  \wh{\bZ}  }_{F}.
\end{align*}
By the second equation in \eqref{perpendicular}, we have ${\rm rank  }( \mathcal{P}_{  \bM^{\ast}   }) \leq 2 {\rm rank}(\bM^{\ast})$. Take $\tilde{\tau} \geq 2^{-1}3\tilde{C}\tilde{\Delta}$, it follows that
\begin{align}\label{rank inequality}
&\quad\frac{1}{nL}\langle   \bR \circ \bP^{\dagger},  \bPi^{\dagger} \circ (\wh{\bZ} - \bZ^{\ast})\circ  (\wh{\bZ} - \bZ^{\ast})    \rangle \notag\\
&\leq \tilde{\tau} \sqrt{2{\rm rank}(\bM^{\ast})} \Norm{  \bZ^{\ast} -  \wh{\bZ}  }_{F} + 
\frac{2}{3}\tilde{\tau} \sqrt{ 2{\rm rank}(\bZ^{\ast}) } \Norm{  \bZ^{\ast} -  \wh{\bZ}  }_{F}\notag\\
&\leq \tilde{\tau} \left\{\sqrt{{\rm rank}(\bM^{\ast})}  + \frac{2}{3}\sqrt{ {\rm rank}(\bZ^{\ast}) }\right\} \Norm{  \bZ^{\ast} -  \wh{\bZ}  }_{F}\notag\\
&\leq \tilde{\tau}\sqrt{2 {\rm rank}(\bM^{\ast}) + \frac{8}{9} {\rm rank}(\bZ^{\ast}) }    \Norm{  \bZ^{\ast} -  \wh{\bZ}  }_{F}\notag\\
&\leq \tilde{\tau}\sqrt{2 D + \frac{26}{9} {\rm rank}(\bZ^{\ast}) }    \Norm{  \bZ^{\ast} -  \wh{\bZ}  }_{F}.
\end{align}
By %Lemma~\ref{projection inequality}, 
Lemma S6 in the supplementary material,
we have
\begin{align*}
\Norm{\wh{\bZ} - \bZ^{\ast}}_{\ast} 
 \leq &\Norm{[\bzero, \wh{\bZ} - \bZ^{\ast}]}_{\ast} \notag\\
 \leq &\Norm{\mathcal{P}_{  \bM^{\ast} } \left( [\bzero, \bZ^{\ast} -  \wh{\bZ} ] \right)}_{\ast} + \Norm{\mathcal{P}_{  (\bM^{\ast})^{\perp} } \left( [\bzero, \bZ^{\ast} -  \wh{\bZ} ] \right)}_{\ast}  \notag\\
 \leq & 6  \Norm{\mathcal{P}_{  \bM^{\ast} } \left( [\bzero, \bZ^{\ast} -  \wh{\bZ} ] \right)}_{\ast}\notag\\
 \leq & \sqrt{72 {\rm rank}(\bM^{\ast})  }  \Norm{ \bZ^{\ast} -  \wh{\bZ}}_{F}.
\end{align*}
Denote $\tilde{\beta} = \norm{ \wh{\bZ} - \bZ^{\ast}  }_{\infty}$. Apparently, $\tilde{\beta} \leq 2 \beta$. Further, define
\begin{align*}
&\mathcal{C}(r)= \left\{  \bZ \in \mathbb{R}^{n \times L}: \norm{\bZ}_{\infty} = 1, \norm{\bZ}_{\ast} \leq \sqrt{r} \norm{\bZ}_{F}, \langle \bPi^{\dagger} ,      (\wh{\bZ} - \bZ^{\ast} ) \circ (\wh{\bZ} - \bZ^{\ast} ) \rangle  \geq \alpha_{U} \sqrt{\frac{64\log(n+L)}{ \log(6/5) nL }}     \right\}.
\end{align*}
\begin{enumerate}
\item If $ \langle \bPi^{\dagger},  (\wh{\bZ} - \bZ^{\ast}) \circ   (\wh{\bZ} - \bZ^{\ast}) \rangle < \alpha_{U} \beta^{2}  \sqrt{\frac{64\log(n+L)}{ \log(6/5) nL }}$, it yields
\begin{align*}
\frac{1}{nL}\Norm{\wh{\bZ} -  \bZ^{\ast}}_{F}^{2}\leq &\frac{1}{\alpha_{L}nL}\langle \bPi^{\dagger},  (\wh{\bZ} - \bZ^{\ast}) \circ   (\wh{\bZ} - \bZ^{\ast}) \rangle
\leq  \frac{\alpha_{L}} {\alpha_{U} \beta^{2}}  \sqrt{\frac{64\log(n+L)}{ \log(6/5) nL }}.
\end{align*}
\item Otherwise, $ \langle \bPi^{\dagger},  (\wh{\bZ} - \bZ^{\ast}) \circ   (\wh{\bZ} - \bZ^{\ast}) \rangle \geq \alpha_{U} \beta^{2}  \sqrt{\frac{64\log(n+L)}{ \log(6/5) nL }}$. In this case, $\tilde{\beta}^{-1} (\wh{\bZ} - \bZ^{\ast}) \in \mathcal{C}( 72 {\rm rank}(\bM^{\ast})  )  $. By \eqref{rank inequality}, 
%Lemma~\ref{phi 3} and Lemma \ref{ minor decomposition }, 
Lemma S4 and Lemma S5 in the supplementary material,
we have 
\begin{align*}
\frac{1}{nL}\Norm{\wh{\bZ} -  \bZ^{\ast}}_{F}^{2} 
\leq & \frac{1}{\alpha_{L}nL}\langle \bPi^{\dagger},  (\wh{\bZ} - \bZ^{\ast}) \circ   (\wh{\bZ} - \bZ^{\ast}) \rangle\notag\\
 \leq&   \frac{C   {\rm rank}({\bM^{\ast}})}{ \alpha_{L}p_{\rm min}} [     \tau^{2} +   \{\mathbb{E}( \norm{\bPhi^{(4)}})\}^{2}]\notag\\
 \leq& \frac{C \tau^{2}   {\rm rank}    ({\bM^{\ast}})   }{\alpha_{L}p_{ \rm min}} + 
 \frac{C \alpha_{U}^{2}  {\rm rank}  ({\bM^{\ast}})    \log(n+L)   }{\alpha_{L} p_{\rm min}(n  \wedge L)},
\end{align*}
holds with probability at least $1- 3/(n+L)$, where $\bPhi^{(4)}$ is defined in Lemma S4
%Lemma~\ref{phi 3}
in the supplementary material, which completes the proof of the theorem.
\end{enumerate}
\end{proof}

Theorem~\ref{statistical guarantee} presents an upper bound with the same order as the methods in \citet{robin2020main} and  \citet{alaya2019collective}, but with a different proof technique faithfully incorporating sampling weights, estimated response probabilities, and the covariate matrix. 
Meanwhile, the following theorem ensures that the proposed algorithm enjoys a sub-linear convergence rate.
\begin{theorem}\label{convergence analysis}
    Suppose Assumptions 1$\sim$5 hold. Then, there exists constant $C_{3} > 0$ such that for $\tau \geq C_{3} \wt{\Delta}$, 
    \begin{align}
     \abs{\mathcal{L}_{\tau} (\bZ_{1}^{(k)} )  - \mathcal{L}_{\tau} (\bZ^{\ast}) }\leq \frac{   \tau \norm{    \bZ_{1}^{(0)} - \bZ^{\ast}        }_{F}^{2}             }{(k+1)^{2}     }
    \end{align}
    holds with probability at least $1 - 3/(n+L)$.
\end{theorem}

\begin{proof}
The proof of Theorem 2 is similar to the proof of Theorem 4.1 in \cite{ji2009accelerated}. We first claim that $h$ is convex. For $\alpha \in(0,1)$, $\bC, \bD  \in  \mathbb{R}^{n \times L}$, we have
\begin{align*}
 h( \alpha \bC + (1-\alpha) \bD  ) &= 
 \Norm{ [\bX, \alpha \bC + (1-\alpha) \bD  ]   }_{\ast}\notag\\
 &= 
 \Norm{ \alpha[\bX,  \bC   ]  + (1 - \alpha)[\bX,  \bD   ]  }_{\ast}\notag\\
 &\leq \alpha  \Norm{ [\bX, \bC ]   }_{\ast} + (1-\alpha)  \Norm{ [\bX, \bD ]   }_{\ast}\notag\\
 &= \alpha  h(\bC) + (1-\alpha) h(\bD).
\end{align*}
We conclude that $\wh{\ell}_{w}\left(\bZ\right)$ is convex as  $(g^{(s)})''(z) > 0$  from Assumption~\ref{bound of exponential family}, so $\mathcal{L}_{\tau} (\bZ)$ is convex.

Secondly, we verify the subgradient form of $h(\bZ) =  \Norm{ [\bX, \bZ ]   }_{\ast}$. The subgradient $\bG$ for $h$ at $\bZ = \bZ_{0}$ indicates that, for any $\bZ \in  \mathbb{R}^{n \times L}$,
\begin{align}\label{orginal subgradient}
 h(\bZ) \geq h(\bZ_{0}) + \langle \bZ - \bZ_{0}, \bG \rangle,
\end{align}
where $\langle \cdot, \cdot \rangle$ is induced by trace norm. Expand \eqref{orginal subgradient}, it yields$
 \Norm{ [\bX, \bZ ]   }_{\ast} -  \Norm{ [\bX, \bZ_{0} ]   }_{\ast}
 \geq  \langle \bZ - \bZ_{0}, \bG \rangle$.
On the other hand, for any $\tilde{\bG} \in \partial\Norm{[\bX, \bZ_{0}]}_{\ast}$, that is, the subgradient for nuclear norm of the $n \times (D + L)$ matrix at $[\bX, \bZ_{0}]$, we have
\begin{align*}
 \Norm{ [\bX, \bZ ]   }_{\ast} -  \Norm{ [\bX, \bZ_{0} ]   }_{\ast} \geq  \langle [\bX, \bZ] - [\bX, \bZ_{0}], \tilde{\bG} \rangle 
 = \langle [\bzero, \bZ - \bZ_{0}] , [\tilde{\bG}_{1}, \tilde{\bG}_{2}] \rangle
= \langle \bZ - \bZ_{0}, \tilde{\bG}_{2} \rangle. 
\end{align*}
We know that the subgradient 
\begin{align*}
\partial\Norm{[\bX, \bZ_{0}]}_{\ast} = &\{ \bU \bV^{\trans} + \bW: \bW \in \mathbb{R}^{n \times (d + L)}, \bU\trans\bW = \bzero, \bW\bV = \bzero, \Norm{\bW} \leq 1      \},
\end{align*}
where $[\bX, \bZ_{0}] = \bU\bSigma\bV$ is an SVD. Therefore, let 
 $\tilde{\bI} = [\bzero_{n\times d}, \bI_{n \times L}]\trans$, 
it turns out that
\begin{align*}
\partial h(\bZ_{0}) = &\{ (\bU \bV^{\trans} + \bW)\tilde{\bI} : \bW \in \mathbb{R}^{n \times (d + L)}, \bU\trans\bW = \bzero,
\bW\bV = \bzero, \Norm{\bW} \leq 1      \},
\end{align*}
where the subgradient of $h$ is verified.
\end{proof}

\section{Synthetic Experiment}

We conduct a simulation study under a stratified two-stage cluster sampling design. Specifically, we assume a finite population consists of $H$ strata,   
 $M_h$ clusters in the $h$-th stratum and $M_{hi}$ elements in the $(hi)$-th cluster.  
Assume stratum effects $ a_h \sim \mathrm{Ex}(1)$ for $h=1,\ldots,H$, where $\mathrm{Ex}(\lambda)$ is an exponential distribution with rate parameter $\lambda$. 
   Generate the stratum sizes $
   M_h\mid a_h \sim 5\mathrm{Po}(a_h) +20$ for $h=1,\ldots,H$, where $\mathrm{Po}(\lambda)$ is a Poisson distribution with parameter $\lambda$.
   Further, generate  cluster effects $b_{hi}\sim\mathrm{Ex}(1)$ for $i=1,\ldots,M_{h}$ and the cluster sizes $
   M_{hi} \sim 5\mathrm{Po}(a_h+b_{hi}) +30$ for $i=1,\ldots,M_h$.
   We consider the following sampling design.
     The first-stage sampling design is probability proportional to size sampling with replacement, where the selection probability is proportional to the cluster size, and the sample size is $m_1$. That is, for $h=1,\ldots,H$, independently sample $m_1$ clusters with selection probability proportion to $M_{hi}$ for $i=1,\ldots,M_h$. 
     The second-stage sampling is simple random sampling without replacement, and the sample size is $m_2$. Then, for $h=1,\ldots,H$, the inclusion probability for each sampled element is $\pi_{hij} = m_1m_2/N_h$ for $i=1,\ldots,M_{h}$ and $j=1,\ldots,M_{hi}$, where $N_h = \sum_{i=1}^{M_h}M_{hi}$ is the stratum size.
     We consider $(m_1,m_2) = (5,20)$ for each setup of $H$.
   The above sampling mechanism results in sample size $n = 900$.

   Within each cluster, we generate an auxiliary matrix and a response matrix of interest.
   Specifically, generate $\bX_{\ast,hi}^{0} \in \mathbb{R}^{M_{hi} \times D}$ whose entries are independent and identically distributed from Ex(1).
   Let $\bX_{\ast,hi} = \bX_{\ast,hi}^{0} / \norm{ \bX_{\ast,hi}^{0}  }_{\infty}$. 
   Further, generate coefficient matrix $\bW_{hi}^{(s)} \in \mathbb{R}^{d \times m_{s}}$ whose elements are independently generated from a uniform distribution over $(0,2)$ for $s = 1, 2, 3$.  Let $\wt{\bZ}_{\ast,hi}^{(s)} = \bX_{\ast,hi}\bW_{hi}^{(s)}$ and $\bZ_{\ast,hi}^{(s)} = \wt{\bZ}_{\ast,hi}^{(s)} / \norm{ \wt{\bZ}_{\ast,hi}^{(s)}  }_{\infty}$. 
   Each $\bY_{hi}^{(s)}$ are generated by parameters in $\bZ_{\ast,hi}^{(s)}$.
   Specifically, we assume that the entries of $\bY_{hi}^{(1)}$ come from the Gaussian distribution, the entries of $\bY_{hi}^{(2)}$ come from the Poisson distribution and the entries of $\bY_{hi}^{(3)}$ come from the Bernoulli distribution. 
\begin{figure}[htb]
   \begin{minipage}[b]{0.48\linewidth}
     \centering
     \centerline{\includegraphics[width=7.0cm]{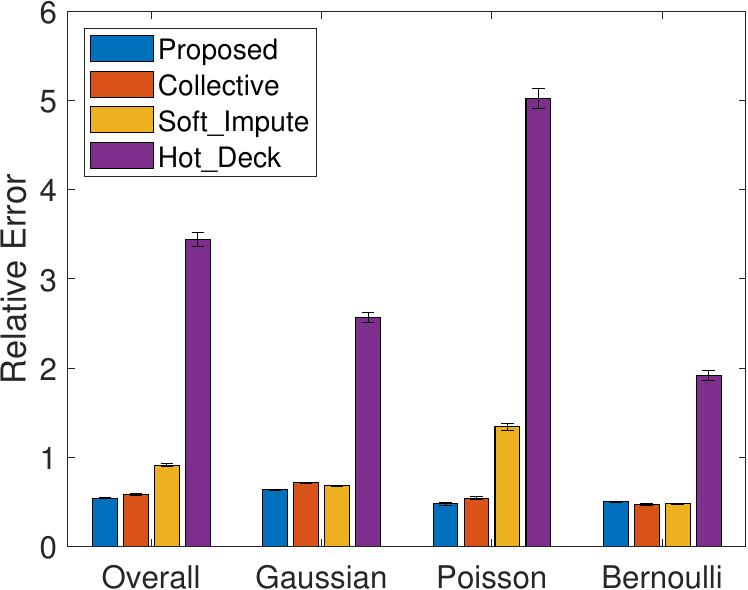}}
   %  \vspace{2.0cm}
     \centerline{(a) response probability 40\%}\medskip
   \end{minipage}
   \hfill
   \begin{minipage}[b]{.48\linewidth}
     \centering
     \centerline{\includegraphics[width=7.0cm]{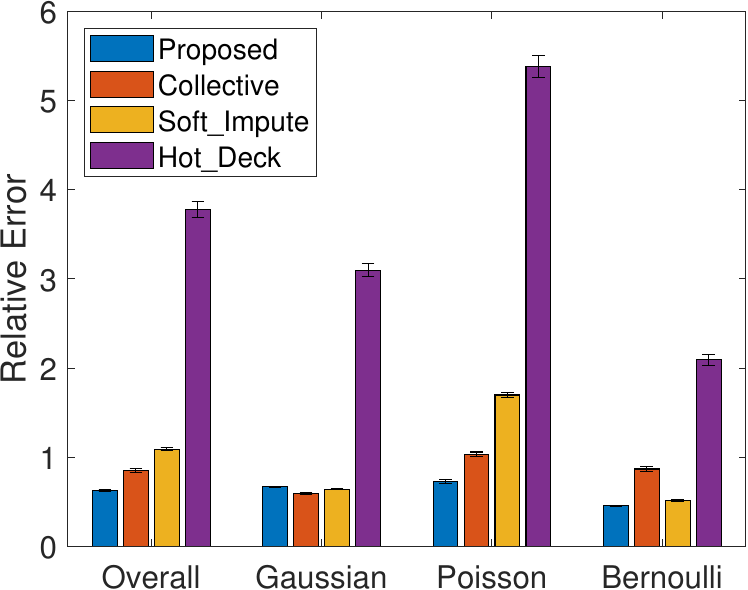}}
   %  \vspace{1.5cm}
     \centerline{(b) response probability  60\%}\medskip
   \end{minipage}

   \caption{ Relative errors for each method with 100 Monte Carlo samples in three response probability scenarios (including hot deck). }
   \label{fig:main}
   \end{figure}

     For entry-wise missing mechanism, let $\bzeta_{j,h}^{(s)} = ( \zeta_{1,j,h}^{(s)} , $  $\zeta_{2,j,h}^{(s)}, \zeta_{3,j,h}^{(s)}, \zeta_{4,j,h}^{(s)}    )\trans$. 
   We generate $\zeta_{1,j,h}^{(s)}$ from $N(\xi, 0.1^{2})$ and $\zeta_{k,j,h}^{(s)}$ from $N(0.3, 0.1^{2})$ for $k = 2, 3, 4$,
    where $\xi$ is used to adjust the overall missing rate. In this simulation study, we take $\xi = 0.3, -0.1, -0.5$ for response probability $40\%$ and $60\%$. 
   In addition, take $m_{1} = m_{2} = m_{3} = 300$. We further choose learning depth $K = 200$. We tune the parameter $\tau$ using an independently generated validation dataset and apply the same parameter to the $100$ Monte Carlo samples. We compare the proposed method with three other popular 
   matrix completion methods:
   \begin{itemize}
      \item \textbf{Collective}: penalized likelihood accelerated inexact soft impute method from collective matrix completion \cite{alaya2019collective}.
      \item \textbf{Soft\_Impute} : Soft-Impute method \cite{mazumder2010spectral}. 
      \item \textbf{Hot\_Deck} : Hot deck imputation method \cite{kim2004fractional}. 
   \end{itemize} 
    To measure the performance of each method, we employ the   relative error of recovered matrix $\wh{\bS}$ compared with true matrix $\bS^{\ast}$, i.e., 
$
      \mbox{RE}(\wh{\bS}, \bS^{\ast}) =  \norm{\wh{\bS} - \bS^{\ast}}_{F}  / \norm{\bS^{\ast}}_{F}.
$
 
 In the simulation, we compare different methods in the aspects of RE($\wh{\bZ}$, $\bZ^{\ast}$) for the overall recovered matrix, RE($\wh{\bZ^{(1)}}$, $\bZ^{\ast(1)}$) in Gaussian response part,  RE($\wh{\bZ^{(2)}}$ $\bZ^{\ast(2)}$) in Poisson response part and 
 RE($\wh{\bZ^{(3)}}$, $\bZ^{\ast(3)}$) in Bernoulli response part. The results are presented in Figure~\ref{fig:main}, where the error bar on top of each colored bar represents one standard error for the corresponding relative error.
 The proposed method has the lowest overall relative error in all three missing scenarios.  More specifically, in 
 each sub-matrix, our proposed method also presents nearly lowest relative errors in RE($\wh{\bZ^{(1)}}$, $\bZ^{\ast(1)}$), RE($\wh{\bZ^{(2)}}$, $\bZ^{\ast(2)}$) and RE($\wh{\bZ^{(3)}}$, $\bZ^{\ast(3)}$). 
 The numerical results partially support the benefit of auxiliary information for heterogeneous missingness for better matrix completion under survey sampling framework.

\begin{comment}
\begin{figure}[t!]
    \centering
    \begin{subfigure}{0.4\textwidth}
           \centering
           \includegraphics[width=6cm]{plot_n_900_num_60_linear.epsc}
            \caption{missing rate 60\%}
            \label{fig:a}
    \end{subfigure}
    \begin{subfigure}{0.4\textwidth}
            \centering
            \includegraphics[width=6cm]{plot_n_900_num_50_linear.epsc}
            \caption{missing rate 50\%}
            \label{fig:b}
    \end{subfigure}
	\begin{subfigure}{0.4\textwidth}
		\includegraphics[width=6cm]{plot_n_900_num_40_linear.epsc}
		\caption{missing rate 40\%}
		\label{fig:c}
\end{subfigure}
    \caption{ Relative errors for each method with 100 Monte Carlo samples in three missing rate scenarios. }
    \label{fig:main}
\end{figure}
\end{comment}

\section{Real Data Application}

The National Health and Nutrition Examination Survey (NHANES) is a comprehensive national survey conducted every two years to provide representative data on the health and nutritional status of adults and children in the United States. NHANES employs a multi-faceted approach to data collection, involving interviews, physical examinations, and laboratory tests. This survey covers a wide range of health-related topics, including demographic information, health conditions, nutrition status, and health behaviors, among others.
The goal is to recover the values in the 2015-2016 NHANES data.

The dataset comprises $n = 5375$ sampled units, each associated with corresponding sampling weights.  The number of strata $H$ in our real data application is 14; see \cite{chen2020national} for more details.  The covariate information includes 16 fully observed demographic variables, such as gender, age and marital status. To create the covariate matrix, we standardize these variables, ensuring that their columns have zero mean and unit variance.
The response matrix consists of answers to 130 survey questions, and a majority of them have missing values. Among these questions, 57 are binary answers and 73 are continuous-valued. We assume that the binary variables follow Bernoulli distributions and the continuous variables follow Gaussian distributions. 
Following the convention  of  \cite{alaya2019collective} and \cite{robin2020main}, we standardize the columns for continuous responses. We revert to the original scale to present the final results.
 We apply the four methods assessed in the previous section to this dataset. To select the appropriate tuning parameter, we adopt a five-fold cross-validation approach. The dataset is randomly divided into five folds, and within each fold, we employ the training set to complete the matrix and calculate the squared Euclidean norm for the difference between the imputed mean and the observed entries in the validation set. The sum of these errors serves as the cross-validation criterion for a specific tuning parameter. We choose the best tuning parameter with the lowest cross-validation error from  $\{2^{-15},\ldots, 2^{-1}, 1, 2\}$.

\begin{table}[ht]
\caption{Mean estimation for the following six questions, where ``I" stands for ``What is the highest grade or level of school {you have/SP has} completed or the highest degree {you have/s/he has} received?",
``II" for ``Total household income (reported as a range value in dollars)",                                                                              
``III" for ``How many of those meals {did you/did SP} get from a fast-food or pizza place?",                                                                      
``IV"  for ``Total savings or cash assets at this time for {you/NAMES OF OTHER FAMILY/your family}.",                                                         
``V"  for ``Because of {your/SP's} (high blood pressure/hypertension), {have you/has s/he} ever been told to $\ldots$ take prescribed medicine?",                   
``VI"  for ``During the past 30 days, how many times did {you/SP} drink {DISPLAY NUMBER} or more drinks of any kind of alcohol in about two hours?".}
\label{tab:real data}
\centering
\begin{tabular}{cccccc}
  \hline
  & Response rate & Collective & Soft\_Impute & Hot\_Deck & Proposed \\ 
  \hline
I    & 0.95 & 3.75 & 3.75 & 3.76 & 3.74 \\ 
II   & 0.93 & 9.88 & 9.89 & 9.91 & 9.89 \\
III  & 0.77 & 1.86 & 1.98 & 2.08 & 1.99 \\ 
IV   & 0.66 & 1.57 & 1.37 & 1.54 & 1.30 \\ 
V    & 0.35 & 0.80 & 0.52 & 0.80 & 0.80 \\ 
VI   & 0.23 & 0.56 & 0.56 & 0.90 & 0.55 \\ 
   \hline
\end{tabular}
\end{table}

The weighted mean for each question is computed using the imputed matrix and the survey weights. To present the results, we have randomly selected six questions, categorized by different levels of response probabilities. Specifically, we have two questions with high response rates of 0.95 and 0.93, two questions with moderate response rates of 0.77 and 0.66, and two questions with low response rates of 0.35 and 0.23. The corresponding results along with descriptions of the questions are provided in Table~\ref{tab:real data}.
Notably, when the response rates are high, the results obtained from the four methods are relatively similar. However, as the response rates decrease, the four methods begin to diverge from each other. It is important to observe that in all scenarios, the results from our proposed method consistently agree with one of the other three methods, demonstrating its robustness and practical use.

\section{Conclusion}

In this work, we study mixed matrix completion in survey sampling with heterogeneous missingness.
Statistical guarantees of the proposed method and the upper bound of the estimation error are presented. We propose an algorithm for estimation, and its convergence analysis shows that it achieves a sub-linear convergence.
Experimental results support our theoretical claims and the proposed estimator shows its merits compared to other existing methods. Some research directions are worthy of investigation. To accommodate longitudinal observations, extension to the tensor completion of the proposed method is a potential research topic. To improve the numerical performance, a non-convex factorization \citep{zhao2015nonconvex} method for $\bZ^{\dagger}$ in \eqref{main object} can be considered.

\section*{Supplementary Materials}

Supplementary materials contain a demo Matlab code for our proposed method and useful lemmas for Theorem~\ref{statistical guarantee}.

\bibliographystyle{chicago}
\bibliography{Survey_Mix_MC_refine}

% \newpage
% \appendix

% \section{Technical Details}
% In this appendix, we provide technical proofs and computational time comparisons for our main article. Specifically, it is organized as follows:
% \begin{itemize}
% 	\item Section \ref{SectionA} introduces necessary preliminary;
% 	\item Section \ref{SectionB} presents proof of Theorem \ref{algorithm convergence};
% 	\item Section \ref{SectionC} presents proof of Theorem \ref{T1};
%     \item Section \ref{SectionD} provides auxiliary lemmas.
%     \item Section \ref{SectionE} provides computational time comparison for the simulations.
% \end{itemize}

% \section{PRELIMINARY}\label{SectionA}

\end{document}